\documentclass[journal]{IEEEtran}

\usepackage{cite}
\usepackage[cmex10]{amsmath}
\usepackage{hyperref}
\usepackage{amssymb}
\usepackage{amsthm}
\usepackage{amssymb,amsthm,amscd}
\newcommand{\F}{\mathbb{F}}
\newcommand\qp{{\frac{q-1}{p-1}}}

\newcommand{\GL}{{\mathrm{GL}}}
\newcommand{\GGL}{{\mathrm{\Gamma L}}}
\newcommand{\Tr}{{\mathrm{Tr}}}

\newtheorem{theorem}{Theorem}
\newtheorem{lemma}{Lemma}
\newtheorem{corollary}{Corollary}
\newtheorem{proposition}{Proposition}
\theoremstyle{definition}
\newtheorem{remark}{Remark}

\begin{document}
\title{\bf A new approach to the Kasami codes of type 2
\thanks{This is the accepted version of the paper published in IEEE Trans. Inf. Theory 66(4) 2020, 2456--2465, 
\url{https://doi.org/10.1109/TIT.2019.2949609} \c IEEE 2019}%
\thanks{This research is supported by
the National Natural Science Foundation of China (61672036), 
the Excellent Youth Foundation of Natural Science Foundation of Anhui Province (1808085J20), 
the Academic Fund for Outstanding Talents in Universities (gxbjZD03), 
and 
the Program of Fundamental Scientific Researches of the Siberian Branch of the Russian Academy of Sciences (proj. No.0314-2019-0016).}%
}
\author{%
Minjia Shi%
\thanks{%
     M. Shi is with the Key Laboratory of Intelligent Computing and Signal Processing of Ministry of Education, School of Mathematics, Anhui University, Anhui, 230601, P. R. China
     (e-mail: smjwcl.good@163.com)%
}
, Denis S. Krotov
\thanks{%
     D. S. Krotov is with the Sobolev Institute of Mathematics,
     pr. Akademika Koptyuga 4,
     Novosibirsk 630090, Russia (e-mail: krotov@math.nsc.ru)%
}
, Patrick Sol\'e
\thanks{%
     P. Sol\'e is with the CNRS/LAGA, University of Paris 8, 2 rue de la Libert\'e, 93 526 Saint-Denis, France (e-mail: sole@enst.fr)
} 
} 

\date{}
\maketitle
\begin{abstract}
\boldmath
The dual of the Kasami code of length $q^2-1$, with $q$ a power of $2$, is constructed by concatenating a cyclic MDS code of length $q+1$ over $F_q$ with a Simplex code of length $q-1$. This yields a new derivation of the weight distribution of the Kasami code, a new description of its coset graph, and a new proof that the Kasami code is completely regular. The automorphism groups of the Kasami code and the related $q$-ary MDS code are determined. New cyclic completely regular codes over finite fields a power of $2$, generalized Kasami codes, are constructed; they have coset graphs isomorphic to that of the Kasami codes. Another wide class of completely regular codes, including additive codes, as well as unrestricted codes, is obtained by combining cosets of the Kasami or generalized Kasami code.
\end{abstract}



\section{Introduction}\label{s:intro}
Distance-regular  graphs form the most extensively studied class of structured graphs due to their many connections with codes, designs, 
groups and orthogonal
polynomials \cite{BanIt-1,Brouwer}. 
Since the times of Delsarte \cite{Delsarte:1973}, 
a powerful way to create distance-regular graphs, 
especially in low diameters, 
{{has been to use}} the coset
graph of completely regular codes. 
A linear or additive code is {\em completely regular} 
if the weight distribution of each coset solely depends on the weight of its coset leader. 
{{Many such examples from Golay codes, Kasami codes, and others, can be found in \cite{Brouwer}}}. 
A recent survey is \cite{BRZ:CR}.

In this note we focus on the Kasami code of type (ii) 
in the sense of \cite[\S~11.2]{Brouwer}. 
This code was also studied in \cite{BRZ:2015}.
It is a cyclic code of length $2^{2m}-1$ with the two zeros $\alpha$, 
$\alpha^{ 1+2^m}$, 
where $\alpha$ denotes a primitive root of $\F_{2^{2m}}$.
We call such a code a classical Kasami code. 
We give a new construction by concatenating a cyclic MDS code with a binary Simplex code.
This gives a new derivation of its weight
distribution, a non trivial calculation in \cite{MWS},
and a characterization of its automorphism group.

More generally, we construct {\em generalized} Kasami codes of lengths $\frac{q^2-1}{p-1}$, 
when $q$ is an even power of $p$, and $p$ itself a power of $2$.
This is obtained by concatenation of a $q$-ary cyclic MDS code with a $p$-ary Simplex code.
The resulting code is cyclic with two zeros $\alpha$, $\alpha^{(p-1)(1+q)}$, 
where $\alpha$ denotes a primitive root of $\F_{q^2}$.
When $p>2$, we obtain in that way infinitely many new cyclic completely
regular codes with a coset graph isomorphic to the coset graph of a Kasami code.

We emphasize that one of the main contributions of our research 
is in applying the concatenation with the Simplex
code to \emph{cyclic} codes, 
with the result being cyclic too. 
In general, 
the concatenation approach in the construction 
of linear completely regular codes
is known, see 
\cite{RifZin:Kronecker}, \cite{RifZin:2017}, 
\cite{BRZ:2018}, \cite{BRZ:CR}.
In particular, in \cite{BRZ:2018}, 
the concatenation of dual-distance-$3$ MDS codes 
with the Simplex code was considered 
(note that a distance-$3$ MDS code, even non-linear, is always completely regular \cite[Cor.~6]{KKO:smallMDS}); the result is a two-weight $q$-ary code, 
whose dual is a covering-radius-$2$ completely regular code.
We use a similar approach, but start with a cyclic dual-distance-$4$ MDS code,
and prove also the cyclicity of the resulting $p$-ary ($p$ is a power of two) 
completely regular code of covering radius $3$, which is known as the Kasami code if $p=2$, 
and is a new cyclic completely regular code if $p>2$.
In fact, concatenation with the Simplex code can be applied as well to produce 
a completely regular $p$-ary code 
from a linear completely regular $q$-ary code, $q=p^k$, with the same intersection array,
independently on the concrete parameters.
Moreover, even more general approach guarantees that
from any unrestricted 
(not necessarily linear or additive)
$q$-ary completely regular code 
(and even an equitable partition)
we can construct a completely regular code 
with the same intersection matrix
over any alphabet $p$ such that $q$ is a power of $p$ 
(the alphabet sizes $p$ and $q$ are not required
to be prime powers).
This approach is based on the existence of 
a locally bijective homomorphism,
known as a \emph{covering}, 
see, e.g., \cite[Sect. 6.8]{GoRo},
between the Hamming graphs $H(n,q)$ 
(in the role of the \emph{cover})
and $H(\frac{q-1}{p-1}n,p)$ 
(in the role of the \emph{target}) 
of the same degree.
With the inverse homomorphism, an equitable partition or
completely regular code of the target graph automatically maps to 
an equitable partition or completely regular code, respectively, of the cover graph, with the same 
intersection matrix.
This combinatorial approach is very powerful and allows to connect equitable partitions of different graphs,
especially translation graphs (Cayley graphs over abelian groups), 
such as, for example, bilinear forms graphs and Hamming graphs.
However, it does not always allow to track 
some nice algebraic properties of the codes such as cyclicity.

The material is arranged as follows.
The next section collects the notions and notations needed for the rest of the paper. 
Section~\ref{s:concat} describes the concatenation approach and the inner code used.
Section~\ref{s:main} gives a new description of the classical Kasami codes, 
including their weight distributions, and their automorphism groups. 
Section~\ref{s:gal} describes the (new) generalized Kasami codes
obtained by allowing a more general alphabet in the outer code. 
This includes a determination of the weight distribution, and of the coset graphs. 
Section~\ref{s:union} provides a combinatorial way to derive
new completely regular codes from old, and applies it to the generalized Kasami codes. 
Section~\ref{s:concl}  recapitulates the path followed, and points out new directions.

\section{Definitions and notation}\label{s:def}

\subsection{Graphs}\label{ss:graphs}
All graphs in this note are finite, undirected, connected, without loops or multiple edges.
In a graph $\Gamma$,
the neighborhood $\Gamma(\bar v)$ of the vertex $\bar v$ is the set of vertices connected to $\bar v$.
The {\em degree} of a vertex $\bar v$ is the size of $\Gamma(\bar v)$.
A graph is {\em  regular} if every vertex has the same degree.
The $i$-neighborhood $\Gamma_i(\bar v)$ is the set of vertices at geodetic distance $i$ to $\bar v$.
A graph is {\em  distance regular} (DR) if for every pair or vertices $\bar u$ and $\bar v$ at distance $i$ apart the quantities
$$
 b_i=| \Gamma_{i+1}(\bar u)\cap \Gamma(\bar v)|, \qquad
 c_i=| \Gamma_{i-1}(\bar u)\cap \Gamma(\bar v)|,
$$
which are referred to as the \emph{intersection numbers} of the graph,
solely depend on $i$ and not on the special choice of the pair $(\bar u,\bar v)$.
The {\em  automorphism group} of a graph is the set of permutations of the vertices that preserve adjacency.
%
 The {\em Hamming graph} $H(n,q)$ (we consider only the case of prime power $q$)
 is a distance-regular graph on $\F_q^n$,
 two vectors being connected
 if they are at Hamming distance one. 
 It is well known \cite{Markov} that
 any automorphism of the Hamming graph acts as a permutation $\pi$ of
 coordinates and, for every coordinate, a permutation $\sigma_i$
 of $\F_q$; i.e., $\sigma\circ\pi$: $(v_1,\ldots,v_n)\to (\sigma_1(v_{\pi^{-1}(1)}),\ldots,\sigma_n(v_{\pi^{-1}(n)}))$.

 \subsection{Equitable partitions, completely regular codes}\label{ss:crc}
{ A partition $(P_0,\ldots,P_{k})$ of the vertex set of a graph $\Gamma$ is called
 an \emph{equitable partition}}
 (also known as regular partition, partition design, perfect coloring)
 if there are constants $S_{ij}$ such that
 $$|\Gamma(\bar v)\cap P_j|=S_{ij}\qquad \mbox{for every } \bar v\in P_i, \qquad i,j=0,\ldots,k.$$
 The numbers $S_{ij}$ are referred to as the intersection numbers,
 and the matrix
 $(S_{ij})_{i,j=0}^{k}$ as the \emph{intersection} (or quotient) \emph{matrix}.
 If the intersection matrix is tridiagonal,
 then $P_0$ is called a \emph{completely regular code} of covering radius $k$ and
 \emph{intersection array} $(S_{01},S_{12},\ldots, S_{k{-}1\, k}; S_{10},S_{21},\ldots, S_{k\, k{-}1})$.
 That is to say, a set of vertices $C$ is a completely regular code
 if the distance partition with respect to $C$ is equitable.
 As was proven by Neumaier \cite{Neumaier92}, for a distance regular graph,
 this definition is equivalent to the original Delsarte definition \cite{Delsarte:1973}:
 a code $C$ is completely regular if
 the outer distance distribution $(|C\cap \Gamma_i(\bar v)|)_{i=0,1,2,...} $ of $C$ with respect to a vertex $\bar v$ depends only on the distance from $\bar v$ to $C$.

\subsection{Linear and additive codes}\label{ss:linear}
A  linear code (that is, a linear subspace of $\F_q^n$) of length $n$, dimension $k$, minimum distance $d$ over the field $\F_q$ is called an $[n,k,d]_q$ code.
An $[n,k,d]_q$ code is called an MDS code if {$d=n-k+1$}.
The duality is understood with respect to the standard inner product.
Let $A_w$ denote the number of codewords of weight $w$. 

Linear codes are special cases of the additive codes, which are, by definition, the subgroups of the additive group of  $\F_q^n$.
A {\em coset} of an additive code $C$ is any translate of $C$ by a constant vector. A {\em coset leader} is any coset element that minimizes the weight.
The {\em weight of a coset} is the weight of any of its leaders.
The {\em coset  graph} $\Gamma_C$ of an additive code $C$
is defined on the cosets of $C$,
two cosets being connected if they differ by a coset of weight one.
\begin{lemma}[see, e.g., \cite{Brouwer}]\label{l:cr-dr}
An additive code
with distance at least $3$ is completely regular 
with intersection array $\{b_0,\ldots,b_{\rho-1};c_1,\ldots,c_{\rho} \}$ if and only if
the coset graph is distance-regular with intersection numbers $b_0,\ldots,b_{\rho-1},c_1,\ldots,c_{\rho}$.
\end{lemma}

The {\em weight distribution} 
of a code is displayed as $[\langle 0,1\rangle,\cdots,\langle w,A_w\rangle,\cdots]$,
where $A_w$ denotes the number of codewords of weight $w$.

\subsection{Cyclic codes}
A cyclic code $C$ of length $n$ over $\F_q$ is, up to polynomial representation of vectors, an ideal of the ring $\F_q[x]/(x^n-1)$. All ideals are principal
with a unique monic generator, called {\em the generator polynomial} of $C$, and
denoted by $g(x)$. If this polynomial has $t$ irreducible factors over $\F_q$, 
then the code $C$ is called a cyclic code {\em with $t$ zeros}.
In that case, the dual $C^\perp$ admits a trace representation with $t$ terms of the form
$$c(a_1,\dots,a_t)=\Big(\sum_{i=1}^t \Tr(a_i \alpha_i^j)\Big)_{j=1}^n  , $$
where
\begin{itemize}
 \item $k_i$ is the degree of the factor indexed $i$ of $g(x)$,
 \item $\Tr=\Tr_{q^{k_i}/q}$ is the trace from $\F_{q^{k_i}}$ down to $\F_q$,
 \item $a_i$ is arbitrary in $\F_{q^{k_i}}$,
 \item $\alpha_i$ is a primitive root in $\F_{q^{k_i}}^\times$,
\end{itemize}

In general the {\em relative trace} from $\F_{q^r}$ to $\F_q$,
{where $q$ is a power of a prime},
is defined  as
$$\Tr_{q^r/q}(z)=z+z^q+\cdots+z^{q^{r-1}}.$$
If $q$ is prime, then $\Tr_{q^r/q}(z)$ is the {\em absolute trace}.
See the theory of the Mattson--Solomon polynomial in \cite{MWS} for details and a proof.

\subsection{Automorphisms}
\newcommand\Aut{\mathrm{Aut}}
\newcommand\AutO{\Aut0}
\newcommand\AutM{\mathrm{MAut}}
\newcommand\AutL{\mathrm{AutL}}
\def\AutM{\mathrm{AutM}}
\def\AutS{\mathrm{AutS}}

The {\em automorphism group} $\Aut(C)$ of a $q$-ary code  $C$ of length $n$ is the stabilizer of $C$ in the automorphism group of the Hamming graph $H(n,q)$.
Note that every automorphism acts as the composition of a coordinate permutation and a permutation of the alphabet in every coordinate.

We will say that a subgroup $A$ of the automorphism group of the Hamming graph $H(n,q)$ acts \emph{imprimitively} on the coordinates
if the set of coordinates can be partitioned into subsets, called \emph{blocks}, and denoted by $N_1$, \ldots, $N_t$, $1<t<n$,
such that the coordinate-permutation part of every automorphism from $A$ sends every block to a block.
If there is no such a partition, then the action of the group on the coordinates is \emph{primitive}
(for example, this is the case for $\Aut(C)$ of any cyclic code $C$ of prime length).

We denote by $\AutO(C)$ the stabilizer of the all-zero word in $\Aut(C)$.
The {\em monomial automorphism group} $\AutM(C)$ of a $q$-ary code  $C$ is defined as the set of all {\em monomial transforms} that leave the code wholly invariant \cite{HuffmanPless}.
A monomial transform acting on $\F_q^n$ is the product of a permutation matrix by an invertible diagonal matrix, both in $\F_q^{n \times n}$. The subgroup consisting of all such transforms with trivial diagonal part is called the {\em permutation part} of the automorphism group.
The group $\AutS(C)$ consists of the 
{\em semilinear automorphisms} of a $q$-ary code $C$,
where each semilinear automorphism is the composition $\tau\circ\mu$
of a monomial transform $\mu$ and an automorphism $\tau$ 
of the field $\F_q$
(acting simultaneously on the values of all coordinates)
such that $\tau\circ\mu(C)=C$.
Note that for any linear code $C$, the groups $\AutM(C)$ and $\AutM(C^\perp)$ are always isomorphic, as well as  the groups $\AutS(C)$ and $\AutS(C^\perp)$; the same is not necessarily true for  $\Aut$ or $\AutO$.

If $q$ is not prime and $\F_q$ has a subfield $\F_p$, 
we also use the notation $\AutL_p(C)$ to denote the subgroup of $\Aut0(C)$ consisting 
of only transformations that are linear over $\F_p$. Any automorphism from $\AutL_p(C)$ acts as the composition of a coordinate permutation and 
a $\F_p$-linear permutation (from $\GL(\log_p q,p)$) 
of $\F_q$ in every coordinate. In particular, if $q$ is a power of $2$,
then
we have $\AutM(C)\le \AutS(C)\le \AutL_2(C) \le \AutO(C) \le \Aut(C)$,
and if $q=2$, then the first four of these groups are the same 
(the last group, for an additive code, 
is the product of $\AutO(C)$ and the group of translations 
by codewords).

\section{Concatenation with the Simplex code}\label{s:concat}

Concatenation is the replacement of the symbols of the codewords of some code,
called the \emph{outer code}, by the codewords of some other code, called the \emph{inner code}.
The role of the inner code in our study is played by a Simplex code;
moreover, we mainly focused on the case when this code is cyclic.

Let $p$ be a prime power; let $q=p^k$.
Let $\xi$  be a root of unity of order $\qp$ in $\F_q$.
We assume that $k$ and $p-1$ are coprime; hence, so are $\frac{q-1}{p-1}$ and $p-1$
(indeed, by induction, $\frac{p^k-1}{p-1}\equiv k\bmod p-1$).
It follows that
\begin{equation}\label{eq:Fq=}
\F_q^\times = \big\{b \xi^l \,|\, b \in \F_p^\times, l\in\{1,\ldots,\textstyle\frac{q-1}{p-1}\}\big\}
\end{equation}
and, in particular, all powers of $\xi$ are mutually linear independent over $\F_p$.
Then, the matrix $(\xi^1,\ldots,\xi^{\frac{q-1}{p-1}})$, considered as a $k\times\frac{q-1}{p-1}$ matrix over $\F_p$, is a generator matrix of a Simplex code of size $q$,
where the distance between any two different codewords is $p^{k-1}$ \cite[Th. 2.7.5]{HuffmanPless}.

\textbf{Encoding the Simplex code.}
Let $L$ be a nondegenerate $\F_p$-linear function from $\F_q$ to $\F_p$.
Any such function can be represented
as $L(\cdot)\equiv \Tr_{q/p}(a\cdot)$ for some
$a \in \F_q$ , and in our context we can always
assume that $L$ is $\Tr_{q/p}$.
Define $\phi:$ $\F_q\to \F_p^\qp$ by
$$ \phi(z) = \big( L(z\xi^1),\ldots ,L(z\xi^\qp) \big); $$
expand its action to $\F_q^n$ coordinatewise,
and to any set of words in $\F_q^n$ wordwise.

\begin{lemma}\label{l:scaled}
 For any two words $\bar y$, $\bar z$ in $\F_q^n$, 
 we have $d(\phi(\bar y),\phi(\bar z))=p^{k-1}d(\bar y,\bar z)$.
 I.e., $\phi$ is a scaled isometry.
\end{lemma}
\begin{proof} 
Immediately from the fact that {the Simplex code} is a one-weight code with nonzero weight $p ^{k-1}$,
we have
\begin{IEEEeqnarray*}{rCl}
 d(\phi(y_1,\ldots,y_n),\phi(z_1,\ldots,z_n)) \hspace{-4cm}\\
 &=&  d \big(  (\phi(y_1),\ldots,\phi(y_n)),(\phi(z_1),\ldots,\phi(z_n)) \big) \\
  &=& \sum_{i=1}^n d(\phi(y_i),\phi(y_i))  
    = \sum_{i=1}^n p ^{k-1} d(y_i,y_i)  \\ & =& p^{k-1}d((y_1,\ldots,y_n),(z_1,\ldots,z_n)).
\end{IEEEeqnarray*}
\end{proof}

\begin{lemma}\label{l:GM}
Assume that {$M=(M_{ij})_{i=1}^m \vphantom{M}_{j=1}^{n}$}
is an $m\times n$ generator matrix of a code $C$ in $\F_q^n$. Then the matrix
$$\Phi(M)=\left(
\begin{array}{c@{\ }c@{\ }c@{\ }c@{\ }c@{\ }c}
M_{11}\xi^1 & ... & M_{11}\xi^\qp & M_{12}\xi^1
& \ldots &  M_{1n}\xi^\qp \\
... & ... & ... & ...
& \ldots &  ... \\
M_{m1}\xi^1 & ... & M_{m1}\xi^\qp & M_{m2}\xi^1
& \ldots &  M_{mn}\xi^\qp
\end{array} \right),$$
considered as a $km \times \qp n$ matrix over $\F_p$,
is a generator  matrix of the code $\phi(C)$
in $\F_p^{\qp n}$.
\end{lemma}

\begin{proof}
The proof is straightforward.
Any codeword $\bar c$ of the code generated by $\Phi(M)$
is a linear combination of its rows, and so is of the form
$$\bar c=\sum_{i=1}^m \big(L_i(M_{i1}\xi^1), \ldots,  L_i(M_{in}\xi^\qp)\big)$$
for some $\F_p$-linear functions $L_i:\F_q\to\F_p$. Since $L_i(\cdot)\equiv L(a_i\cdot)$ for some $a_i\in \F_q$, we have
\begin{eqnarray*}
 \bar c&=&\sum_{i=1}^m \big({L(a_iM_{i1}\xi^1)}, \ldots,  L(a_iM_{in}\xi^\qp)\big) 
 \\&=&
\sum_{i=1}^m \big(\phi(a_i M_{i1}),\ldots,\phi(a_i M_{in})\big)\\
&=&
\sum_{i=1}^m \big(\phi(a_i M_{i1},\ldots,a_i M_{in})\big)\\
&=&
\phi \Big(\sum_{i=1}^m a_i(M_{i1},\ldots,M_{in})\Big).
\end{eqnarray*}
That is, $\bar c$ is a $\phi$-image
of some codeword of $C$.
\end{proof}

The following lemma provides an alternative explanation to the fact
that completely regular codes with identical coset graphs can exist 
over different alphabets \cite{RifZin:2017}.

\begin{lemma}\label{l:coset-graph}
 The coset graph $\Gamma_{C^\perp}$ of the dual $C^\perp$ of a linear code $C$ in $\F_q^n$ 
 and the coset graph $\Gamma_{\phi(C)^\perp}$ of the dual $\phi(C)^\perp$ of its image 
 $\phi(C)$ in $\F_p^{\qp n}$ are isomorphic.
\end{lemma}
\begin{proof}
 We first observe that the number of cosets of $C^\perp$ coincides with size of the code $C$,
 which is the same as the size of the code $\phi(C)$,
 which is, in its turn, the number of cosets of  $\phi(C)^\perp$. 
 That is, the considered graphs have the same number of vertices.

 Let $C$ be a code over $\F_q$ with a generator matrix $M$, 
 which is a check 
 matrix of the dual code $C^\perp$.
 The cosets of $C^\perp$ naturally correspond
 to the syndromes under the action of the check matrix $M$.
 By the definition of the coset graph $\Gamma_{C^\perp}$, 
 two cosets are adjacent in the coset graph
 if and only if the difference
 between the corresponding syndromes is
 the syndrome of a weight-$1$ word.
 We call the syndromes of the weight-$1$ words the \emph{connecting syndromes}
 because they play a connecting role for the coset graph
 (this terminology is in agree with that for a more general class of graphs called Cayley graphs). 
 For the check matrix {$M=(M_{ij})_{i=1}^m \vphantom{M}_{j=1}^{n}$} over
 $\F_q$, the connecting syndromes
 are 
 \begin{equation}\label{eq:syn-a}
a(M_{1j},\ldots,M_{mj})^{\mathrm T},\quad j\in \{1,\ldots,n\}, \quad a\in \F_q^\times
 \end{equation}
 (corresponding to the weight-$1$ words with $a$ in the $j$-{th} position).
 
 Now, consider the matrix $\Phi(M)$,
 which, by Lemma~\ref{l:GM},
 is a generator matrix for $\phi(C)$
 and a check matrix for $\phi(C)^\perp$.
 The connecting syndromes are 
  \begin{IEEEeqnarray}{l}\label{eq:syn-b}
 b(\xi^l M_{1j},\ldots,\xi^l M_{mj})^{\mathrm T}, 
 \\ 
  \nonumber
 \qquad
 j\in \{1,\ldots,n\},
 \quad l\in \Big\{1,\ldots,\qp\Big\},
 \quad b\in \F_p^\times,  
  \end{IEEEeqnarray}
 corresponding to the weight-$1$
 words with $b$
 in the position $(j-1)\qp +l$.
  Since $\{ b\xi^l \mid b\in \F_p^\times,\ l\in \{1,\ldots,\qp\}\}$
  coincides with $\F_q^\times$ (see \eqref{eq:Fq=}), 
  we find that the sets of connecting syndromes \eqref{eq:syn-a}, \eqref{eq:syn-b} 
  for both matrices coincide; call this set $S$.
  
  Two cosets of $C^\perp$ are adjacent in $\Gamma_{C^\perp}$ 
  if and only if the corresponding syndromes have the difference in $S$.
    Two cosets of $\phi(C)^\perp$ are adjacent in $\Gamma_{\phi(C)^\perp}$ 
  if and only if the corresponding syndromes have the difference in $S$.
  So, indexing the cosets by the syndromes
  induces an isomorphism between the $\Gamma_{C^\perp}$ and $\Gamma_{\phi(C)^\perp}$.
 \end{proof}

\begin{corollary}\label{c:crc-ia}
If the code $C^\perp$ is completely regular,
then the code $\phi(C)^\perp$
is completely regular with the same intersection array.
\end{corollary}

In the rest of this section, we establish important relations 
between the automorphism groups of a code and its concatenated image.

\subsection{Automorphism group}

\begin{lemma}\label{l:aut-to-aut}
The group $\AutL_p(C)$ is isomorphic to a subgroup
of $\AutM(\phi(C))$ such that the action of this subgroup 
on the coordinates is imprimitive with the blocks
$$\Big\{1,\ldots ,\qp\Big\},\ \ldots,\ \Big\{\qp(n{-}1){+}1,\ldots ,\qp n\Big\}.$$
\end{lemma}
\begin{proof}
The action of the monomial automorphism group on a codeword $c$ of $C$ can be represented as a permutation of the coordinates and, for every $i$ from $1$ to $n$
the action of an element of $\GL(m,p)$ on of the value 
in the $i$-{th} coordinate, understood as an element of $\F_p^m$.
We first consider these two actions separately.

(i) A permutation of coordinates corresponds to a permutation of blocks for $\phi(C)$:
we see that
$$ 
\phi(z_{\pi^{-1}(1)},
\ldots,
z_{\pi^{-n}(n)}) 
=(\phi(z_{\pi^{-1}(1)}),
\ldots,
\phi(z_{\pi^{-n}(n)})) $$
is obtained from $\phi(z_{1},\ldots,z_{n})$ by sending each $i$-{th}
block of coordinates to the place of the $\phi(i)$-{th} block.

(ii) Next, we show that the action of a 
$\F_p$-linear permutation
on the value $z_i$ in the 
$i$-{th} coordinate corresponds to
a monomial transform inside the $i$-{th} block of
$\phi(z_{1},\ldots,z_{n})$.
Here, a crucial observation if the following well-known fact.

(*) \emph{Any bijective linear transform of the Simplex code is equivalent to the action 
of some of its monomial automorphisms.} {\small This fact is straightforward from the form of 
a generator matrix of the Simplex code, which consists of a complete collection
of pair-wise linearly independent columns of height $m$. Applying a bijective linear transform
to the rows, we obtain another generator matrix with the same property.
It is obvious that one such matrix can be obtained from another one by permuting 
 columns and multiplying each column by a nonzero coefficient, 
that is, by a monomial transform.}

Consider the value $z_i$ in the $i$-{th} position 
and some permutation $\tau_i$ of $\F_q$
that is linear over $\F_p$. Since $\phi$ is also linear over $\F_p$ and invertible,
we have a bijective linear transform $ \phi(\tau_i(\phi^{-1}(\cdot))) $ of the Simplex code. 
By (*),
there is a monomial automorphism $\sigma_i$ of the Simplex code whose action on this code 
is identical 
to $ \phi(\tau_i(\phi^{-1}(\cdot))) $. 
Then, we have
$$  \phi(\tau_i(z_i)) =  \sigma_i(\phi(z_i)), $$
i.e., the action of a $\F_p$-linear permutation $\tau_i$ in the $i$-{th} coordinate 
of a word $\bar z$ from $\F_q^n$
corresponds to the action of the monomial transform $\sigma_i$ in the $i$-{th} block
of coordinates of $\phi(\bar z)$.

Summarizing (i) and (ii), we get the following.
If the mapping
$$
(z_1,\ldots,z_n) \to
\left(\tau_i(z_{\pi^{-1}(1)}),\ldots,\tau_i(z_{\pi^{-n}(n)})\right) $$
belongs to $\AutL_p(C)$, then
$$
(\phi(z_1),\ldots,\phi(z_n)) \to
\left(\sigma_i(\phi(z_{\pi^{-1}(1)})),\ldots,\sigma_i(\phi(z_{\pi^{-n}(n)}))\right) $$
maps $\phi(C)$ to itself, i.e., belongs to $\AutM(\phi(C))$.
This proves the statement.
\end{proof}

\begin{lemma}\label{l:subgroup}
Assume that $C$ is a linear code over 
$\F_q$ and $C^\perp$ 
satisfies the following property: 
\begin{enumerate}
    \item[\rm (i)] the minimum distance of $C^\perp$ is at least $4$.
\end{enumerate}
Then the groups  $\AutL_p(C)$, $\AutM(\phi(C))$, 
and $\AutM(\phi(C)^\perp)$ are isomorphic.
\end{lemma}

\begin{proof}
We will first deduce from (i) that
\begin{enumerate}
    \item[(ii)] 
    \emph{$\AutO(\phi(C)^\perp)$ (and hence, also $\AutM(\phi(C)^\perp)$)
    acts imprimitively on the coordinates
    with the \emph{blocks}
$$\Big\{1,\ldots ,\qp\Big\},\ \ldots,\ \Big\{\qp(n-1)+1,\,\ldots ,\,\qp n\Big\}.$$}
\end{enumerate}
(As we see from Lemma~\ref{l:aut-to-aut}, this property
is indeed necessary, while (i) is only sufficient to guarantee (ii).)
To show (ii), we note two other facts. From the code distance $\ge 4$ of $C^\perp$, we see that
\begin{enumerate}
    \item[(iii)] \emph{for every codeword of weight $3$ in $\phi(C)^\perp$, 
    the three nonzero coordinates
    belong to the same block.} 
    {\small Indeed, assume that $\phi(C)^\perp$ has a codeword with
    exactly three nonzero coordinates, with values $\beta$, $\beta'$ and $\beta''$ respectively.
    This means that for some $i$, $i'$, $i''$, $j$, $j'$, $j''$ such that
    $(i,j)\ne (i',j')\ne (i'',j'')\ne (i,j)$,
    the equation
    $$ \beta L(z_{i} \xi^{j} ) + \beta' L(z_{i'} \xi^{j'} ) + \beta'' L(z_{i''} \xi^{j''} ) =0$$
    is satisfied for every $(z_1,\ldots,z_n)$ from $C$. Since $L$ is linear and nondegenerate,
    this implies
    $$ (\beta \xi^{j}) z_{i} + (\beta' \xi^{j'}) z_{i'} + (\beta'' \xi^{j''}) z_{i''}=0$$
    for all $(z_1,\ldots,z_n)$ from $C$.
    The last equation contradicts to the code distance of $C^\perp$, 
    unless $i=i'=i''$, i.e., the three coordinates
    are from the same block.} 
\end{enumerate}
On the other hand, because the dual of a Simplex code is a $1$-perfect Hamming code, we have
\begin{enumerate}
    \item[(iv)] \emph{every two coordinates
    from the same block are nonzero coordinates for
    some weight-$3$ codeword in $\phi(C)^\perp$.}
    {\small Indeed, {in the generator matrix of the Simplex code}, 
    the sum of any two columns is proportional to some other column.
    The same relation takes place 
    between the corresponding coordinates 
    from the same block of the concatenated code,
    which implies the required property.}
\end{enumerate}
From (iii) and (iv), we see that the code properties of $\phi(C)^\perp$ 
distinguish the pairs of coordinates in the same block and the pairs of coordinates in different blocks. 
Hence, the automorphism group does not break blocks and we have (ii).

Next, we observe that property (i) excludes some degenerate cases with vanishing coordinates:

\begin{enumerate}
    \item[(v)] \emph{every codeword of the Simplex code occurs in every block, 
    over the codewords of the concatenated code
    $\phi(C)$.}
    {\small Indeed, this obviously equivalent to the fact that for every symbol of $\F_q$ 
    and every position, there is a codeword of $C$ with the given symbol in the given position.
    Since $C$ is linear, it is the same as to say that there is no vanishing coordinate in $C$.
    Equivalently, $C^\perp$ does not contain a weight-$1$ codeword,
    which follows immediately from (i).
    }
\end{enumerate}

Now, we have that every automorphism 
from $\AutM(\phi(C)^\perp)$ 
permutes blocks,
permutes coordinates inside {each block},
and permutes the values of $\F_p$ for every coordinate
(fixing $0$).
The same is true for $\AutM(\phi(C))$, because it is isomorphic to $\AutM(\phi(C)^\perp)$
with the isomorphism keeping the permutation part.
According to (v), 
the last two steps (permuting coordinates inside {each block},
and permuting the values of $\F_p$ for every coordinate)
must be compatible with
the automorphism group of the Simplex code,
and result in a permutation of its codewords.
A permutation of the codewords
of the Simplex code in each block,
for the code $\phi(C)$,
corresponds to a permutation of the elements of $\F_q$
in the corresponding coordinate for $C$;
and a permutation of the blocks corresponds
to a permutation of the coordinates for $C$.
So, we see that every automorphism from $\AutM(\phi(C))$
corresponds to some automorphism of $C$:
$$ \AutM(\phi(C)) \lesssim \AutO(C).  $$
Moreover, because all considered transformations are linear over $\F_p$, 
we have
$$ \AutM(\phi(C)) \lesssim \AutL_p(C). $$
From Lemma~\ref{l:aut-to-aut},
we have the inverse correspondence.
It remains to note that $\AutM$
of a code and its dual are isomorphic.
\end{proof}

\section{Classical Kasami codes}\label{s:main}
Let $q=2^m$ for some integer $m\ge 1$.
Consider the cyclic code $M_q^\bot$
of length $q+1$ defined over $\F_q$
by the check polynomial $h(x)=(x-1)(x-\zeta)(x-\zeta^q)$,
for the root $\zeta=\alpha^{q-1}$ of order
$q+1$ over $\F_{q^2}$, where $\alpha$ is a primitive root of $\F_{q^2}$ (for some notation reasons, we first 
define $M_q^\bot$ and then $M_q$ as the dual to $M_q^\bot$).
Note that $\zeta^{q^2}=\zeta$,
so that 
$$(x-\zeta)(x-\zeta^q)=x^2-(\zeta+\zeta^q)x+\zeta^{q+1}=x^2-(\Tr_{q^2/q}\zeta)x+1 $$ 
has its coefficients in $\F_q$.
\begin{theorem} \label{th:mds}
The code $M_q^\bot$ is MDS of parameters $[q+1,3,q-1]_q$. Its weight distribution is
\begin{equation}\label{eq:mds}
 \Big[
\big\langle 0,1\big\rangle ,
\Big\langle q-1,{\frac{q^3-q}{2}}\Big\rangle ,
\big\langle q,q^2-1\big\rangle ,
\Big\langle q+1,\frac{q^3-2q^2+q}{2}\Big\rangle
\Big]. 
\end{equation}
\end{theorem}

\begin{proof}
 The BCH bound (see e.g. \cite[\S\,7.6]{MWS}) applied to $M_q$ shows that $M_q$, hence $M_q^\bot$ is MDS. The frequency of weight $q-1$ is derived on applying \cite[Corollary~5, p.\,320]{MWS}.
 $$A_{q-1}=(q-1)
 \binom{q+1}{q-1}
 =(q-1)
 \binom{q+1 }{ 2}
 . $$
  The frequency of weight $q$ is derived on applying \cite[p.\,320]{MWS}.
  $$A_q=
  \binom{q+1}{3-2}
  \big[(q^2-1)-q(q-1)\big]=(q+1)(q-1). $$
  The frequency of weight $q+1$ is derived by complementation
  $A_{q+1}=q^3-1-A_{q-1}-A_{q}. $\end{proof}

Denote by $S_q$ the binary {Simplex} code of parameters $[q-1,m,\frac{q}{2}]_2$.
The main result of this section is the following theorem.

\begin{theorem}  \label{th:2}
The concatenation of $M_q^\bot$ with $S_q$ is a binary cyclic code 
$K_q^\perp$ of parameters $[q^2-1, 3m]$, 
with nonzeros $\alpha$, $\alpha^{1+q}$, with $\alpha$ a primitive root of $\F_{q^2}$. 
\end{theorem}

\begin{proof}
As $\zeta=\alpha^{q-1}$ is a root of order $q+1$ of $\F_{q^2}$,
by the facts of Section~2.4, we have
$$ M_q^\bot=\big\{\left(c+\Tr_{q^2/q}(\delta\zeta^j)\right)_{j=1}^{q+1} \,\big|\, c\in \F_q,\, \delta\in \F_{q^2}\big\}. $$
Let $\theta=\alpha^{1+q}$. Clearly $\theta$ is a root of order $q-1$ of $\F_q$.
The concatenation map $\phi$ from $\F_q$ to $S_q$
is conveniently described by using the trace function $\Tr_{q/2}$.
For all $\beta \in \F_q$ let
\begin{IEEEeqnarray*}{rCl}
 \phi(\beta)&=&\big(\Tr_{q/2}(\beta\theta),\Tr_{q/2}(\beta\theta^2),\dots,\Tr_{q/2}(\beta\theta^{q-1})\big)\\
 &=&\big(\Tr_{q/2}(\beta\theta^i)\big)_{i=1}^{q-1}. 
\end{IEEEeqnarray*}
Replacing $\beta$ by the generic coordinate of $M_q^\bot$, we obtain an expression depending on $i$ and $j$.
$$ \Tr_{q/2}(c\theta^i)+\Tr_{q^2/2}(\delta \alpha^{i(q+1)+j(q-1)}  ). $$
By using the Chinese Remainder Theorem for integers for the product $q^2-1=(q+1)\times (q-1)$,
we see that the exponent of $\alpha$  in the second sum ranges over $1,\dots,q^2-1$ modulo $q^2-1$.
This is the Trace formula for a cyclic code with nonzeros $\alpha$, $\theta=\alpha^{1+q}$.
\end{proof}

\begin{corollary}
 The weight distribution of $K_q^\perp$ is
\begin{multline*}
 \Big[
\big\langle 0,1\big\rangle ,\ 
\Big\langle \frac{q^2-q}{2},q\frac{(q^2-1)}{2}\Big\rangle ,\  \\
\Big\langle \frac{q^2}{2},q^2-1\Big\rangle ,\ 
\Big\langle \frac{q^2+q}{2},\frac{q^3-2q^2+q}{2}\Big\rangle
\Big]. 
\end{multline*}
\end{corollary}

\begin{proof}
Concatenation with $S_q$ multiplies the weights of $M_q^\bot$ by a factor $\frac{q}{2}$, 
while leaving their frequencies unchanged.
\end{proof}

\begin{remark}
This result is derived in \cite[Th.~32, p.\,252]{MWS} 
by the complex formula of \cite[Ch.~6, Th.~2]{MWS}.
\end{remark}

\begin{corollary}
 The coset graphs of $K_q$ and $M_q$ are isomorphic. 
In particular \cite{BRZ:2015}, $K_q$ is a completely regular code of diameter $3$ and intersection array
\begin{equation}\label{eq:iaq}
\{q^2 - 1, q(q - 1), 1; 1, q, q^2 - 1\}.
\end{equation}
\end{corollary}

\begin{proof}
The first claim holds by Lemma~\ref{l:coset-graph}, 
where $C=M_q$, $C^\perp=M_q^\bot$,
$\phi(C)=K_q^\perp$, and $\phi(C)^\perp=K_q$.
Since $M_q$ is a completely regular code with intersection array 
\eqref{eq:iaq}, see \cite[F.5]{BRZ:CR},
the coset graphs are distance regular and the code $K_q$ 
is completely regular with the same intersection array. 
\end{proof}

In the rest of this section, we find the automorphism groups of $K_q$ and related codes.

\begin{remark}
The concatenation construction for the special case $q=4$ of the code $K_q$ 
was suggested in \cite[Sect.~4.1]{BRZ:2018},
without considering the cyclicity.
\end{remark}

\subsection{Automorphism groups}\label{s:AutK}

According to Lemma~\ref{l:subgroup}, to know the automorphism group of $K_q$,
we have to find $\AutL_2(M_q^\bot)$. We first prove that it coincides with $\AutS(M_q^\bot)$.

\begin{lemma}\label{l:A(M)}
The automorphism groups of the linear $[q+1,3,q-1]_q$ MDS code $M_q^\bot$ satisfy
$$\Aut0(M_q^\bot)= \AutL_2(M_q^\bot) =\AutS(M_q^\bot).$$
\end{lemma}
\begin{proof}
The proof applies ideas from \cite{Gor:2010} for different code parameters.
Assume that $\pi$ is a permutation of coordinates and 
$\sigma=(\sigma_1,\ldots,\sigma_n)$ is a collection
of permutations of $\F_q$ fixing $0$
such that $\sigma\circ \pi :(v_1,\ldots,v_n) \to
(\sigma_1(v_{\pi^{-1}(1)}),\ldots, \sigma_n(v_{\pi^{-1}(n)}) )$ 
belongs
to $\AutO(M_q^\bot)$. 
Denote by $\pi M_q^\bot$ the code whose codewords are obtained
from the codewords of $M_q^\bot$
by applying the coordinate permutation $\pi$.
So, $\sigma(\pi M_q^\bot)=M_q^\bot$.

 We split the proof into four steps.
The first step is sufficient to see the first equality
 in the claim of the lemma.

(i) 
\emph{We state that for every coordinate $i$ and every $a,b\in \F_q$
it holds $\sigma_i(a+b)=\sigma_i(a)+\sigma_i(b)$.
Equivalently, $\sigma\circ \pi \in \AutL_2(M_q^\bot)$.}
Indeed, consider two codewords $\bar y,\bar z\in \pi M_q^\bot$ of minimum weight, $q-1$,
with zeros in different positions and having $a$ and $b$ in the $i$-{th} position.
Consider the two words $\bar u=\sigma(\bar y+\bar z)$ and $\bar v=\sigma(\bar y)+\sigma(\bar z)$.
They both lie in $M_q^\bot$. 
Moreover, they coincide in every position where $\bar y$ or $\bar z$ has $0$.
Since $\bar y$ has $0$ in two positions and $\bar z$ has $0$ in another pair of positions,
$\bar u$ or $\bar v$ differ in less than $q-1$ positions. As the code distance is $q-1$,
we have $\bar u=\bar v$. In particular, their values $\sigma_i(a+b)$ and $\sigma_i(a)+\sigma_i(b)$ 
in the $i$-{th} coordinate coincide, which proves (i).

(ii) 
\emph{We state that for any  coordinate $i$
and for any values
$a$ and $b$ from $\F_q$, it holds
\begin{equation}\label{eq:propor}
 {\sigma_1(a)}{\sigma_i(b)}=
   {\sigma_1(1)}{\sigma_i(a b)}.
\end{equation}
}
We first assume that $i=2$ and 
there is a weight-$(q-1)$ codeword $\bar y=(1,b,y_3,\ldots,y_n)$ of $\pi M_q^\bot$.
Since the codes $M_q^\bot$ and $\pi M_q^\bot$ are linear and $\sigma\circ\pi$ is an automorphism of $M_q^\bot$,
the words
\begin{IEEEeqnarray*}{rCl}
\sigma (\bar y) 
&=&\left(\sigma_1(1),\sigma_2(b),\sigma_3(y_{3}),\ldots, \sigma_n(y_{n}) \right) \quad \mbox{and}\\
\sigma (a\bar y)
&=&\left(\sigma_1(a),\sigma_2(ab),\sigma_3(a y_{3}),\ldots, \sigma_n(a y_{n}) \right)
\end{IEEEeqnarray*}
are codewords of $M_q^\bot$, 
as well as the words $\bar u=\sigma_1(a) \sigma(\bar y)$ and $\bar v=\sigma_1(1) \sigma(a \bar y)$.
The codewords $\bar u$ and $\bar v$ both have weight $0$ or $q-1$  and the same nonzero positions;
moreover, they coincide in the first coordinate 
(with the value $\sigma_1(1)\sigma_1(a)$).
It follows that the distance between $\bar u$ and $\bar v$ is smaller than the minimum distance
of $M_q^\bot$, and hence $\bar u = \bar v$. In particular, the values $\sigma_1(a)\sigma_2(b)$ 
and $\sigma_1(1)\sigma_2(ab)$ in the second position of $\bar u$ and $\bar v$ coincide, 
which proves \eqref{eq:propor}.

In the next step, we still consider $i=2$.
A straightforward and well-known property of an MDS code of dimension $k$ (and distance $n-k+1$)
is that there is exactly one codeword having chosen values in chosen $k$ distinct coordinates.
By $\bar y^{(l)}$, $l=2,\ldots,n-1$,
we denote the unique codeword of $\pi M_q^\bot$ 
with $1$ in the first coordinate
and $0$ in the $l$-{th} and $n$-{th} coordinates.
For different $l$ and $t$, the words $\bar y^{(l)}$ and $\bar y^{(t)}$ are different
(indeed, $\bar y^{(l)}$ cannot have $0$ in the $t$-{th} position
because its weight cannot be less than the code distance $q-1$).
Hence, they are different in all positions except the first and the last
(the distance cannot be less than $q-1$).
So, $\bar y^{(l)}$, $l=2,\ldots,n-1$, possess $n-2=q-1$ different values 
in the second position. Therefore, for any given $a$,
\eqref{eq:propor} holds for $q-1$ distinct values $b\in F_q$.
 It immediately follows from the bijectivity of $\sigma_i$ that 
 \eqref{eq:propor} holds for the remaining value of $b$.
 So, \eqref{eq:propor} is true for $i=2$ for any $a$ and $b$.
Similarly, it is true for any $i\ne 1$.

The remaining case $i=1$ is derived from the case $i\ne 1$. 
From \eqref{eq:propor} we see that the ratio 
$\displaystyle\frac{\sigma_i(ba)}{\sigma_i(b)}$ 
does not depend on $b$.
By similarity, this holds for any $i$ including $i=1$. So, we have
$\displaystyle\frac{\sigma_1(ba)}{\sigma_1(b)}=\frac{\sigma_1(a)}{\sigma_1(1)}$, 
which is \eqref{eq:propor} with $i=1$. (i) is proven.

(iii) \emph{We state that there is an integer degree $d$ such that
\begin{equation}\label{eq:pow}
\sigma_i(a)=\sigma_i(1)a^d
\end{equation}
for every $i$ from $1$ to $n$ and every $a$ in $\F_q$.}

Indeed, assume $\beta$ is a primitive element.
From \eqref{eq:propor}
we have
$
\sigma_i(a b)=\sigma_i(a)(\sigma_1(b)/\sigma_1(1))
$.
Substituting $b=\beta$, $a=\beta^k$
and 
denoting $d=\log_{\beta} (\sigma_1(\beta)/\sigma_1(1))$,
we get 
$\sigma_i(\beta^{k+1})=\sigma_i(\beta^{k})\cdot \beta ^d$,
which proves \eqref{eq:pow} by induction on $k$.

(iv) \emph{The mapping $a\to a^d$, where $d$ is from (iii), 
is an automorphism of the field $\F_q$.} 
Indeed, this mapping preserves the multiplication:
$(a\cdot b)^d = a^d \cdot b^d$. By (i), it preserves the addition as well.

So, the action of every automorphism $\sigma\circ\pi$ from $\AutO(M_q^\bot)$ can be represented as the combination
of a coordinate permutation, applying an automorphism of $\F_q$ to the values in all coordinates,
and multiplying the values of every coordinate by a constant. By the definition,
such automorphism belongs to $\AutS(M_q^\bot)$.
\end{proof}

To find $\AutS(M_q^\bot)$, we consider the isomorphic group $\AutS(M_q)$.
The next two lemmas
establish the tight upper and lower bounds for this group, respectively.

\begin{lemma}\label{l:AutMqperp}
The semilinear automorphism group $\AutS(M_q)$ of the code $M_q$ 
coincides with $\AutO(M_q)$ (and, in particular, with the intermediate $\AutL_2(M_q)$)
and is isomorphic to a subgroup of the general semilinear group $\GGL(2,q)$. 
\end{lemma}

\begin{proof}
Denote by $H_q$ the Hamming $[q+1,q-1,3]_q$ code 
defined by the generator polynomial $h(x)=(x-\zeta)(x-\zeta^q)$.
Note that $M_q$ is a distance-$4$ subcode of $H_q$.
From the intersection array \eqref{eq:iaq}
of the completely regular code
$M_q$, we find that the number
of vertices at distance at least $3$ from $M_q$
is 
$$ |M_q|\cdot\frac{q^2-1}{1}\cdot\frac{q(q-1)}{q}\cdot\frac{1}{q^2-1}=(q-1)|M_q|,$$
i.e., coincides with $|H_q \backslash M_q|$.
It follows that $H_q$
is a unique distance-$3$ code of size $q^{q-1}$ that includes $M_q$.
It follows that any automorphism $\pi$ from $\AutO(M_q)$ also belongs to $\AutO(H_q)$ 
(otherwise, $\pi(H_q)$ is another distance-$3$ code including $M_q$).
It is known that $\AutS(H_q)$ is isomorphic to $\GGL(2,q)$ \cite[Thm 7.2]{Huffman98}
and, moreover, $\AutO(H_q)=\AutS(H_q)$ \cite{Gor:2010}.
So, all automorphisms from $\AutO(M_q)$ are also semilinear, 
and we get
$
\AutO(M_q)=\AutS(M_q) \le \AutO(H_q)=\AutS(H_q)\simeq \GGL(2,q)
$.
\end{proof}

In \cite[Prop. 3.3]{BRZ:2015} (where their $C^{(u)}$ is our $K_q$),
it is shown that $\AutO(K_q)$ contains a subgroup isomorphic
to the general linear group $\GL(2,q)$. 
The proof of the following lemma pirates arguments from \cite{BRZ:2015},
adopting them to the $q$-ary code $M_q$, instead of the binary $K_q$,
and extending from $\GL(2,q)$ to  $\GGL(2,q)$.

\begin{lemma}\label{l:AutMqperp_}
The semilinear automorphism group $\AutS(M_q)$ of the code $M_q$ 
contains a subgroup of size $|\GGL(2,q)|$. 
\end{lemma}

\begin{proof}
(i) In the first part of the proof, we consider different equations defining $M_q$.
 For each $i$, we express $\zeta^i$ as $\gamma_{i,0}  + \gamma_{i,1} \alpha$,
where $\gamma_{i,0}, \gamma_{i,1} \in \F_q$.
The code $M_q$ consists of all $\bar c=(c_{1},\ldots,c_{q+1})$ from $\F_q^{q+1}$
 such that
 \begin{IEEEeqnarray}{rCl}
 \label{eq:c}
 \sum_{i=1}^{q+1} c_i        &=& 0, 
 \\
 \label{eq:cg}
 \sum_{i=1}^{q+1} c_i \zeta^i &=& 0, \mbox{ equivalently, }
 \sum_{i=1}^{q+1} c_i \gamma_{i,j}=0,\ j=0,1.
\end{IEEEeqnarray}
Assuming that \eqref{eq:cg} holds, we express the left part of \eqref{eq:c} as follows.
\begin{IEEEeqnarray*}{rCl}
 \Big(\sum_{i=1}^{q+1} c_i\Big)^2
 &=&
 \sum_{i=1}^{q+1} c_i^2 (\zeta^i)^0
 =
 \sum_{i=1}^{q+1} c_i^2 (\zeta^i)(\zeta^i)^q \\
 &=&
 \sum_{i=1}^{q+1} c_i^2 (\gamma_{i,0}  + \gamma_{i,1} \alpha)(\gamma_{i,0}  + \gamma_{i,1}\alpha)^q \\
 &=& 
 \sum_{i=1}^{q+1} c_i^2 (\gamma_{i,0}  + \gamma_{i,1} \alpha)(\gamma_{i,0}  + \gamma_{i,1}\alpha^q) \\
 &=& 
 \sum_{i=1}^{q+1} c_i^2 \gamma_{i,0}^2 + 
 \alpha^{q+1} \sum_{i=1}^{q+1} c_i^2 \gamma_{i,1}^2
 \\ &&  \quad 
 {} +  (\alpha+\alpha^{q}) \sum_{i=1}^{q+1} c_i^2 \gamma_{i,0}\gamma_{i,1}\\
 &=& 
 \Big(\sum_{i=1}^{q+1} c_i \gamma_{i,0}\Big)^2 + 
 \alpha^{q+1} \Big(\sum_{i=1}^{q+1} c_i \gamma_{i,1} \Big)^2
 \\ &&  \quad 
 {} +  (\alpha+\alpha^{q}) \sum_{i=1}^{q+1} c_i^2 \gamma_{i,0}\gamma_{i,1} \\
 &=& (\alpha+\alpha^{q}) \sum_{i=1}^{q+1} c_i^2 \gamma_{i,0}\gamma_{i,1}.
\end{IEEEeqnarray*}
Therefore, in the system \eqref{eq:c}--\eqref{eq:cg}, the equation \eqref{eq:c} can be replaced by
 \begin{IEEEeqnarray}{rCl}
 \label{eq:cgg}
 \sum_{i=1}^{q+1} c_i^2 \gamma_{i,0}\gamma_{i,1}        &=& 0.
\end{IEEEeqnarray}
The same arguments can be applied if we replace $\alpha$
by the primitive element $\alpha^t$, where $t$ is a power of two. 
So, \eqref{eq:cgg} can be replaced by 
 \begin{IEEEeqnarray}{rCl}
 \label{eq:cggt}
 \sum_{i=1}^{q+1} c_i^2 \gamma_{i,0}^{(t)}\gamma_{i,1}^{(t)}        &=& 0,
\end{IEEEeqnarray}
where $\gamma_{i,0}^{(t)}, \gamma_{i,1}^{(t)} \in \F_q$
are the coefficients in the expansion
$\zeta^i = \gamma_{i,0}^{(t)}  + \gamma_{i,1}^{(t)} \alpha^t$.

(ii) Now, consider an arbitrary nonsingular semilinear transform
\begin{IEEEeqnarray*}{rCl}
\Psi:\gamma_0 + \gamma_1 \alpha &\to& 
\big((a\gamma_0 + a'\gamma_1)+(b\gamma_0 + b'\gamma_1) \alpha \big)^t
\end{IEEEeqnarray*}
of $\F_{q^2}$ as a two-dimensional vector space over $\F_{q}$,
where $a$, $a'$, $b$, $b'$ are coefficients from $\F_q$ such that 
$\det\binom{a\ a'}{b\ b'}=ab'+a'b\ne 0$ and $t$ is a power of two.
So, $(\cdot)^t$ is an automorphism of $\F_q$,
and $\Psi\in\GGL(2,q)$.
Denote by $\bar e^{(i)}$, $i=1,\ldots,q+1$, the weight-$1$ vector in $\F_{q}^{q+1}$ with $1$ in the $i$-{th} coordinate.
Consider the mapping $\psi$ that maps
$\bar e^{(i)}$ to $d_i \bar e^{(j_i)}$, $d_i\in \F_q$ if $\Psi(\zeta^i) = d_i \zeta^{j_i}$
(by \eqref{eq:Fq=}, every nonzero element of $\F_{q^2}$
is uniquely represented as $d  \zeta^{l}$
for some $l\in\{1,\ldots,q+1\}$ and $d\in\F_q^\times$). 
Using the semilinearity identity
$\psi(\lambda \bar y + \mu  \bar z) = \lambda^t \psi(\bar y)  + \mu^t \psi(\bar z)$, 
we expand the domain of $\psi$ to the whole vector space $\F_q^{q+1}$. 
We will show that $\psi$ lies in $\AutS(M_q)$.

Consider a codeword  $\bar z=(z_{1},\ldots,z_{q+1}) \in M_q$.
It satisfies \eqref{eq:cg} and \eqref{eq:cgg} with $\bar c=\bar z$.
We have to show that \eqref{eq:cg} and \eqref{eq:cggt} hold for $\bar c=\psi(\bar z)$.
We start with  \eqref{eq:cg}.
Note that in the case $\bar c=\psi(\bar e^{(i)})$, the expression 
$\sum_{j=1}^{q+1} c_j \zeta^i$ turns to 
\begin{IEEEeqnarray*}{rCl}
\sum_{j=1}^{q+1} c_j \zeta^j &=&  d_i \zeta^{j_i} = \Psi(\zeta^{i})
\\&=&
\big((a\gamma_{i,0} + a'\gamma_{i,1})+(b\gamma_{i,0} + b'\gamma_{i,1}) \alpha \big)^t .
\end{IEEEeqnarray*}
So, for $$\bar c=\psi(\bar z)=\psi \big(\sum_{i=1}^{q+1} z_i \bar e^{(i)}\big)= \sum_{i=1}^{q+1} z_i^t \psi(\bar e^{(i)})$$ we have
\begin{IEEEeqnarray*}{rCl}
\sum_{j=1}^{q+1} c_j \zeta^j 
&=& \sum_{i=1}^{q+1} \Psi(z_i\zeta^i)  
\\ &=&
\sum_{i=1}^{q+1} z_i^t \big ((a\gamma_{i,0} + a'\gamma_{i,1})+(b\gamma_{i,0} + b'\gamma_{i,1}) \alpha \big )^t
\\ &=&
\Big( (a+ b \alpha) \sum_{i=1}^{q+1} z_i \gamma_{i,0} 
    + (a'+ b' \alpha) \sum_{i=1}^{q+1} z_i \gamma_{i,1} 
    \Big)^t =0
\end{IEEEeqnarray*}
(the last two sums are zeros because of \eqref{eq:cg} for $\bar c=\bar z$), 
which is exactly \eqref{eq:cg}.

Next, deal with \eqref{eq:cggt}.
If $\bar c=\psi(\bar e^{(i)})$,
then the only nonzero element of $\bar c$
is the $j_i$-{th} element,
with the value $c_{j_i}=d_i$,
and 
$d_i \zeta^{j_i} 
= c_{j_i} \gamma_{j_i,0}^{(t)} + c_{j_i} \gamma_{j_i,1}^{(t)} \alpha ^t 
= (a\gamma_{i,0} + a'\gamma_{i,1})^t +(b\gamma_{i,0} + b'\gamma_{i,1})^t \alpha ^t $.
Hence, for $\bar c=\psi(\bar e^{(i)})$ we have
\begin{IEEEeqnarray*}{rCl}
\sum_{j=1}^{q+1} c_j \gamma_{j,0}^{(t)}\cdot c_j\gamma_{j,1}^{(t)} = 
(a\gamma_{i,0} + a'\gamma_{i,1})^t(b\gamma_{i,0} + b'\gamma_{i,1})^t .
\end{IEEEeqnarray*}
Therefore, for $\bar c=\psi(\bar z)$ we have
\begin{IEEEeqnarray*}{rCl}
\sum_{j=1}^{q+1} c_j \gamma_{j,0}^{(t)}\cdot c_j\gamma_{j,1}^{(t)} 
\hspace{-15mm}
\\
&=&
\sum_{i=1}^{q+1} z_i^t (a\gamma_{i,0} + a'\gamma_{i,1})^t  \cdot z_i^t (b\gamma_{i,0} + b'\gamma_{i,1})^t 
\\ &=& 
a^tb^t \Big( \sum_{i=1}^{q+1}  z_i \gamma_{i,0} \Big)^{2t}   +
a'^tb'^t \Big( \sum_{i=1}^{q+1}  z_i \gamma_{i,1} \Big)^{2t} \\
&& \qquad\qquad
{}+ (ab'+a'b)^t  \Big( \sum_{i=1}^{q+1}  z_i^2 \gamma_{i,0}\gamma_{i,1} \Big)^{t}=0
\end{IEEEeqnarray*}
(the last three sums are zeros because of \eqref{eq:cg} and \eqref{eq:cgg} for $\bar c=\bar z$), 
which is exactly \eqref{eq:cggt}.
Clearly, different $\Psi$ correspond to different $\psi$; so, the number of semilinear automorphisms
of $M_q$ is at least $|\GGL(2,q)|$.
\end{proof}

Summarizing the results above, we find the automorphism groups of the Kasami and related codes.

 \begin{theorem}\label{th:AutGL} 
 All the automorphism groups
 $\AutO$, $\AutL_2$, $\AutS$ 
 of the codes $M_q^\bot$, $M_q$, 
 $K_q$, $K_q^\perp$ are isomorphic to each other and to
$\GGL(2,q)$. The groups $\AutM(M_q^\bot)$ and $\AutM(M_q)$ are isomorphic  to
$\GL(2,q)$.
 \end{theorem}

\begin{proof}
For the binary codes, we have
$\AutO(K_q)=\AutL_2(K_q)=\AutS(K_q)=\AutM(K_q)$
and
$\AutO(K_q^\perp)=\AutL_2(K_q^\perp)=\AutS(K_q^\perp)=\AutM(K_q^\perp)$.
By Lemma~\ref{l:A(M)}, 
$\AutS(M_q^\bot)=\AutL_2(M_q^\bot)=\AutO(M_q^\bot)$.
By Lemmas~\ref{l:AutMqperp} and \ref{l:AutMqperp_},
$\AutS(M_q)=\AutL_2(M_q)=\AutO(M_q)\cong \GGL(2,q)$.
By Lemma~\ref{l:subgroup}, $\AutM(K_q^\perp) \cong \AutL_2(M_q^\bot)$.
For dual codes we have
 $\AutS(K_q)=\AutS(K_q^\perp)$,  $\AutS(M_q^\bot)=\AutS(M_q)$, and  $\AutM(M_q^\bot)=\AutM(M_q)$.
By the arguments in the proof of Lemma~\ref{l:AutMqperp}, 
the subgroup $\AutM(M_q)$ of $\AutS(M_q)$ coincides with the subgroup 
$\AutM(H_q)$ of $\AutS(H_q)$ and so isomorphic
to $\GL(2,q)$.
\end{proof}

\section{Generalized Kasami codes}\label{s:gal}
Let $q=p^m$ for some integer $m\ge 1$, and some {\em power of two} $p$. Thus the preceding section is the special case $p=2$ of the present section.
As in the previous section,
we consider the cyclic code $M_q^\bot$
of length $q+1$ defined over $\F_q$
by the check polynomial
$h(x)=(x-1)(x-\zeta)(x-\zeta^q)$, 
for some root $\zeta$ of order
$q+1$ over $\F_{q^2}$.
By~Theorem~\ref{th:mds},
$M_q^\bot$ is a $[q+1,3,q-1]_q$ MDS code
with weight distribution \eqref{eq:mds}.

Denote by $S_q$ the $p$-ary {Simplex code} of parameters $[\frac{q-1}{p-1},m,\frac{q}{p}]_p$.
This code is cyclic when $m$ and $p-1$ are coprime.
The main result of this section is the following theorem.
\begin{theorem}
The concatenation of  $M_q^\bot$ with $S_q$ 
is a $p$-ary code ${K_q^p}^\perp$ 
of parameters $[\frac{q^2-1}{p-1}, 3m]_p$,
with nonzeros $\alpha$, $\alpha^{(1+q)(p-1)}$, 
with $\alpha$ a primitive root of $\F_{q^2}$. 
If $m$ and $p-1$ are coprime, then this code is cyclic. 
\end{theorem}

\begin{proof}
The proof is analogous to that of Theorem~\ref{th:2}. The only difference is the definition of $\theta=\alpha^{(1+q)(p-1)}$, which is now a root of order $\frac{q-1}{p-1}$ in $\F_{q^2}$.
\end{proof}

When $p>2$, as far as we know, the codes $K_q^p$ are new, 
and will be called generalized Kasami codes in this paper.

\begin{corollary}
 The weight distribution of ${K_q^p}^\perp$ is
\begin{multline*}
 \Big[
\big\langle 0,1\big\rangle , \ 
\Big\langle \frac{q^2-q}{p},q\frac{(q^2-1)}{2}\Big\rangle , \ \\
\Big\langle \frac{q^2}{p},q^2-1\Big\rangle , \ 
\Big\langle \frac{q^2+q}{p},\frac{q^3-2q^2+q}{2}\Big\rangle
\Big]. 
\end{multline*}
\end{corollary}

\begin{proof}
Concatenation with $S_q$ multiplies the weights of $M_q^\bot$ by a factor $\frac{q}{p}$, 
while leaving their frequencies unchanged.
\end{proof}

\begin{corollary}
  The coset graphs of $K_q^p$ and $M_q$ are identical. 
In particular, when $p$ is even, $K_q^p$ 
is completely regular of diameter $3$, 
and intersection array 
$$\{q^2 - 1, q(q - 1), 1; 1, q, q^2 - 1\}.$$ 
\end{corollary}

\begin{proof}
This follows by Corollary~\ref{c:crc-ia}.
\end{proof}

\begin{remark}
 We conjecture that $\AutS(K_q^p)$ and $\Aut0(K_q^p)$
 are isomorphic to $\GGL(2,q)$, 
 because it is so for $M_q=K_q^q$ and $K_q=K_q^2$.
 However, our current tools do not allow to prove this easily. 
\end{remark}

\section{Union of cosets}\label{s:union}

In this section, we prove a simple but important fact that cosets of
a linear completely regular code from the classes considered in this paper
can be merged to result in another completely regular code with different
parameters, which can be linear, additive, or even non-additive.

\begin{proposition}
\label{l:union}
Let $p$, $q$ be powers of $2$, and let $C$ be
an additive completely regular code over $\F_p$
with intersection array
$\{ p^{2m} - 1 , p^{2m}- q , 1 ; 1 , q , p^{2m}-1 \}.$
Then
\begin{itemize}
 \item[\rm (i)] for every $k$ from $1$ to $\frac{p^{2m}}{q}-1$, there is
a completely regular code $B_k$
with intersection array
$$I_k=\{ p^{2m} - 1 , p^{2m}- kq , 1 ; 1 , kq , p^{2m}-1 \},$$
$B_k$ being the union of $k$ cosets of $C$;
\item[\rm (ii)] moreover,
if $k$ is a power of $2$, then there is an additive code
with intersection array $I_k$.
\end{itemize}
\end{proposition}

\begin{proof}
(i)
 Denote by $C^{(d)}$
 the set of vertices at distance $d$ from $C$, $d=0,1,2,3$.
 By the definition of a completely regular code,
 $(C^{(0)},C^{(1)},C^{(2)},C^{(3)})$
 is an equitable partition of $H(\frac{p^{2m}-1}{p-1},p)$ with the intersection matrix
 \begin{equation}\label{eq:mx}
  \left(
 \begin{array}{c@{\ }c@{\ }c@{\ }c}
s_{00} & s_{01} & s_{02} & s_{03} \\
s_{10} & s_{11} & s_{12} & s_{13} \\
s_{20} & s_{21} & s_{22} & s_{23} \\
s_{30} & s_{31} & s_{32} & s_{33}
 \end{array}
 \right)
 =
   \left(
 \begin{array}{cc@{\,}cc}
0 &  p^{2m}-1 & 0          & 0 \\
1 &       q-2 & p^{2m}-q   & 0 \\
0 &         q & p^{2m}-q-2          & 1 \\
0 &         0 & p^{2m}-1 & 0
 \end{array}
 \right).
 \end{equation}
By standard double-counting arguments, we have
$|C^{(3)}|/|C^{(0)}|= |C^{(2)}|/|C^{(1)}|= (p^{2m}-q)/q = r-1$, where $r=p^{2m}/q$.
Let $C_0=C^{(0)}=C$ and let $\{C_1,\ldots,C_{r - 1}\}$ be a partition
of $C^{(3)}$ into cosets of $C^{(0)}$.
Each of the cosets of $C$ is a completely regular code with the same parameters. For each $i$ from $0$ to $r - 1$, by $C_{r+i}$ we denote the set of vertices at distance $1$ from $C_i$.
In particular, $C^{(1)}=C_r$ and $C^{(2)}=\cup_{i=1}^{r-1}C_{r+i}$.

We state that $(C_i)_{i=0}^{2r-1}$ is an equitable partition with the quotient matrix
$$
\left(
\begin{array}{cccc|c@{}ccc}
 0 & 0 & \cdots & 0 & p^{2m}{-}1 & 0 & \cdots & 0  \\
 0 & 0 & \cdots & 0 & 0 & p^{2m}{-}1 & \cdots & 0  \\
 \vdots & \vdots & \ddots & \vdots & \vdots & \vdots & \ddots & \vdots \\
 0 & 0 & \cdots & 0 & 0 & 0 & \cdots & p^{2m}{-}1  \\ \hline
 1 & 0 & \cdots & 0 & q{-}2 & q & \cdots & q  \\
 0 & 1 & \cdots & 0 & q & q{-}2 & \cdots & q  \\
 \vdots & \vdots & \ddots & \vdots & \vdots & \vdots & \ddots & \vdots \\
 0 & 0 & \cdots & 1 & q & q & \cdots & q{-}2  \\ 
\end{array}
\right).
$$
To show this, we first note that it is a partition; this is straightforward from the fact that $C_0$, \ldots, $C_{r-1}$ are at distance $3$ from each other (from $s_{33}=0$ and $s_{23}=1$ in \eqref{eq:mx}). Next, we see that all the neighbors of
$C_i$, $i=0,\ldots,r-1$ are in $C_{r+i}$. 
Finally, every vertex of $C_{r+i}$ has one neighbor
from $C_{{i}}$ (by the definition),
exactly $s_{11}=q-2$ neighbors from $C_{r+i}$
(because $C_{{i}}$
is completely regular with intersection array \eqref{eq:mx}),
and exactly $s_{21}=q$ neighbors
from $C_{r+j}$
for every $j$ from
$\{0,\ldots,r-1\}\backslash \{i\}$
(because $C_{{j}}$ is completely regular
with intersection array \eqref{eq:mx}).

Now, it is straightforward to check that denoting
\begin{IEEEeqnarray*}{rClrCl}
B^{(0)}&=&\bigcup_{i=0}^{k-1}C_{i},\qquad &
B^{(3)}&=&\bigcup_{i=k}^{r-1}C_i, \\
B^{(1)}&=&\bigcup_{i=0}^{k-1}C_{r+i},\qquad &
B^{(2)}&=&\bigcup_{i=k}^{r-1}C_{r+i},
\end{IEEEeqnarray*}
we obtain an equitable partition $(B^{(0)},B^{(1)},B^{(2)},B^{(3)})$ with the intersection matrix
$$
\left(
 \begin{array}{cccc}
0 &    p^{2m} - 1 & 0          & 0 \\
1 &           kq-2 & p^{2m}-kq   & 0 \\
0 &             kq & p^{2m}-kq-2 & 1 \\
0 &             0 & p^{2m} - 1 & 0
 \end{array}
 \right).
$$
Therefore, the set $B_k=B^{(0)}$ is a completely regular code with the required intersection array.

(ii) We first note that the set $B_2$ constructed as in (i) is automatically additive. Then, we see that $B_2$ satisfies the hypothesis
of the proposition for the code $C$, with $2q$ instead of $q$.
So, repeating the construction $\log_2 k$ times, we get an additive completely regular code with  the required parameters.
\end{proof}

In the case $m=1$, the code $C$ from Proposition~\ref{l:union} can be chosen
as a {$\frac{1}{p}$-{th} part} of the $p$-ary Hamming code of length $p+1$ in the sense of \cite[(F.5)]{BRZ:CR} (note that $p$ is a power of $2$), in particular, as the code $M_q$,
in notation of Section~\ref{s:main}.

In \cite{BRZ:CR}, this code is classified as belonging to the group of completely regular codes obtained from perfect codes.
Proposition~\ref{l:union} shows that not only {$\frac{1}{p}$-{th} part}, but
also {$\frac{1}{2^i}$-{th} part}, $i=1,2,\ldots,\log_2 p$, of the $p$-ary Hamming code of length $p+1$ can be chosen to be an additive completely regular code. This remark essentially extends the class of completely regular codes obtained from perfect codes.

\section{Conclusion and open problems}\label{s:concl}

In this article, we have used a concatenation scheme to produce new completely regular codes from known ones.
The use of a one-weight code as the inner code is essential in preserving the coset graph structure from the outer code to concatenated code, as evidenced in the proof of Lemma~\ref{l:coset-graph}.
The {Simplex code} is essentially the only one-weight code in the Hamming graph by a classical result of Bonisoli~\cite{Bonisoli}.
It is worth investigating systematically if another choice of the outer code than the cyclic MDS code used here can lead to new constructions of completely regular codes.
More generally, replacing the Hamming graph by other distance regular graphs with a larger choice of one-weight codes could lead to new constructions of completely regular codes
in these graphs.

\subsection*{Acknowledgments}\label{s:ack}

The authors thank Evgeny Gorkunov for useful discussions 
and the anonymous referees for valuable suggestions 
towards improving the paper.


\providecommand\href[2]{#2} \providecommand\url[1]{\href{#1}{#1}}
  \def\DOI#1{{\small {DOI}:
  \href{http://dx.doi.org/#1}{#1}}}\def\DOIURL#1#2{{\small{DOI}:
  \href{http://dx.doi.org/#2}{#1}}}

\end{document}